\documentclass[11pt]{article}
\usepackage[small,compact]{titlesec}
\usepackage{amsmath,amssymb,amsfonts,mathabx,setspace}
\usepackage{graphicx,caption,epsfig,subfigure,epsfig,wrapfig,color}
\usepackage[T1]{fontenc}

\usepackage{setspace}

\usepackage{natbib}

\newtheorem{theorem}{Theorem}

\newtheorem{definition}[theorem]{Definition}

\newtheorem{proposition}[theorem]{Proposition}

\newenvironment{proof}[1][Proof]{\noindent\textbf{#1.} }{\ \rule{0.5em}{0.5em}}

\numberwithin{equation}{section}
\numberwithin{theorem}{section}

\def\Lop{\mathcal{L}}
\def\Rop{\mathcal{R}}
\def\t{\alpha}
\def\e{\epsilon}
\def\R{\mathbb{R}}                            

\def\EXP{\mathbb{E}}
\def\VAR{\mathbb{V}\text{ar}}
\def\Dom{{\rm Dom}}

\newcommand{\ud}[1]{\, \mathrm{d}#1}
\newcommand{\deriv}[3][]{\frac{\ud^{#1} \hspace{-0.3mm} #2}{\ud{#3}^{#1}}}
\newcommand{\pderiv}[2]{\frac{\d{#1}}{\d{#2}}}
\renewcommand{\div}[1][]{\nabla_{\!\! #1}\cdot \!}
\newcommand{\grad}{\nabla }
\renewcommand{\d}{\partial}

\newcommand{\jac}{\boldsymbol{\nabla}}

\begin{document}

\begin{center}
\textbf{\Large Probabilistic Measures for Biological Adaptation and Resilience\footnote{{\it Corresponding Author:} Juan M. Restrepo, {\tt restrepojm@ornl,gov}} 
} \\[0pt]
\vspace{4mm} Jorge M. Ramirez$^1$, Juan M. Restrepo$^{1,2}$, Valerio Lucarini$^{3,4}$, David Weston$^1$\\
$^1$ Oak Ridge National Laboratory, Oak Ridge TN 37831 USA \\
$^2$ Department of Mathematics, University of Tennessee, Knoxville, Knoxville TN 37912 USA \\
$^3$ Department of Mathematics and Statistics, 
University of Reading, Reading UK \\
$^4$ Centre for the Mathematics of Planet Earth, 
University of Reading, Reading UK 
Last updated:  \today \\
\end{center}

\begin{abstract}
This paper introduces a novel approach to quantifying ecological resilience in biological systems, particularly focusing on noisy systems responding to episodic disturbances with sudden adaptations. Incorporating concepts from non-equilibrium statistical mechanics, we propose a measure termed `ecological resilience through adaptation,' specifically tailored to noisy, forced systems that undergo physiological adaptation in the face of stressful environmental changes. Randomness plays a key role, accounting for model uncertainty and the inherent variability in the dynamical response among components of biological systems. Our measure of resilience is rooted in the probabilistic description of states within these systems, and is defined in terms of the dynamics of the ensemble average of a model-specific observable quantifying success or well-being. Our approach utilizes stochastic linear response theory to compute how the expected success of a system, originally in statistical equilibrium, dynamically changes in response to a environmental perturbation and a subsequent adaptation. The resulting mathematical derivations allow for the estimation of resilience in terms of ensemble averages of simulated or experimental data. Finally, through a simple but clear conceptual example, we illustrate how our resilience measure can be interpreted and compared to other existing frameworks in the literature. The methodology is general but inspired by applications in plant systems, with the potential for broader application to complex biological processes.

\end{abstract}


\noindent {\bf Data Availability:} the data is computer generated. Codes are available from the Authors, upon request. Acknowledgement of this paper is a necessary condition for use of the codes that generated the results.\\

\section* {Acknowledgements}
Research was supported by Laboratory Directed Research and Development Program of Oak Ridge National Laboratory, managed by UT-Battelle, LLC, for the U.S. Department of Energy under contract no. DE-AC05–00OR22725. The Department of Energy will provide public access to these results of federally sponsored research in accordance with the DOE Public Access Plan (http://energy.gov/downloads/doe-public-access-plan).
 VL acknowledges the
support received from the Horizon 2020 project TiPES (Grant No. 820970), from the EPSRC project EP/T018178/1, and from the University of Reading's RETF project CROPS.
 
\noindent {\bf Conflict of Interest Statement:} The authors have no conflict of interest. 

\noindent {\bf Authors Contributions:} Jorge Ramirez and Juan Restrepo conceived the ideas, designed and executed the methodology. David Weston proposed the biological problem and provided critical biological context and interpretation. Valerio Lucarini provided expertise in linear response theory and contributed to the interpretation of the results; Juan Restrepo and Jorge Ramirez led the writing of the manuscript. All authors contributed critically to the drafts and gave final approval for publication.

\newpage
\section{Introduction}

Resilience refers to the ability of a system to maintain a certain degree of functionality in the face of disturbances \citep{Dakos22}.  In the seminal work by \citet{Holling73, holling96}, resilience is conceptually categorized as `engineering' or `ecological'. In engineering resilience, perturbations are small and the system returns to an original equilibrium state. Ecological resilience, on the other hand, applies to systems with multiple equilibria and perturbations that can induce the system to move between different attractors. Our focus is ecological resilience, in particular applied to biological systems that operate out of equilibrium and that respond by sudden adaptations as a result of episodic disturbances.

Multiple quantitative resilience measures have been proposed for living systems based upon dynamical systems theory, including linear stability of equilibria, return times, attractor size and geometry, distance to bifurcation manifolds, elasticity and histeresis in response to different types of perturbations: pulses, presses, ramps, deterministic and stochastic. See  \citet{angeler2016quantifying,yi2021review,van2021unifying,Dakos22} for useful reviews. These measures receive often interchangeabe names as resilience, resistance, robustness, stability, recovery, malleability, and tolerance. Altogether, these quantities can help understand how living systems adjust, recover or heal after perturbations, and account for the degree to which they return back to an unperturbed state or to transition to another metastable but viable operating point. 

The mathematical framework of preference for quantifying resilience has been the theory of dynamical systems as in \citet{guck83,Krakovska16,arnoldi2016resilience}, with random perturbations by \cite{Meyer16} and \cite{arani2021exit}, or in networked systems by  \cite{shang2023matrix, shang2023resilient}. Here we argue that ecological resilience can be better achieved by conceptualizing the systems of interest as noisy/forced systems, applying the techniques that occupy the attention of non-equilibrium statistical mechanics. Such a probabilistic framework will be conducive to a notion of ecological resilience that emphasizes persistence and reflects the opportunistic and unpredictable aspects of change in biological systems. It will measure the degree to which a system can adapt to a perturbation by controlling its behavior in order to increase the probabilities of not falling into a state of low well-being. 

Our argument is based on the observation that the temporal evolution of observable biological quantities is rich with noise as a result of the inherent uncertainty in observables and the variability among individuals, or more generally, subsystems. It also incorporates the epistemic error associated with incomplete or uncertain model parameterizations.  Therefore a system’s ability to adjust its function to disturbances can be conceptualized in terms of the effect of the perturbation over the probability distribution of states and the expected value of some measure of well-being. Employing a probabilistic description of these observables allows for a more nuanced understanding of these systems, particularly in describing the living states and behavior that extend beyond the realm of deterministic dynamics.

We propose a measure of \textit{ecological resilience through adaptation} applicable to noisy, forced systems undergoing environmental perturbations. The goal is to provide a mathematical framework in which the dynamic response of the system to perturbation can be quantified, either from ensembles of models or data realizations. Specifically, we provide a probabilistic measure of the degree to which a system in statistical equilibrium can dynamically adapt, recover, or change equilibrium distribution after an stressful perturbation on the parameters modeling its environmental conditions. Here, we use the word `adaptation' in the physiological sense, namely as the ability of living systems to adjust their dynamics (e.g. metabolism) in response to its changing environment. The degree of `recovery' is conceptualized in terms of a \textit{success function} of the state variables designed to model performance, health, productivity; it is an arbitrary measure of well-being that the modeler seeks to maintain as high as possible in average for any given environmental conditions. The probabilistic approach allows, not only for dynamic models of aleatoric, epistemic or measurement uncertainty (e.g. observations and parameters in stochastic differential equations), but also to assess resilience of an aggregate or ensemble of subsystems exhibiting random variations on their state variables. 

Our formulation uses modern methods in stochastic process theory to capture changes in the probabilistic distribution of the state of a forced/noisy systems that is dynamically perturbed. Specifically, we rely on \textit{linear response theory}  to predict how the expected value of the success function changes in response to environmental disturbances and physiological adaptations. This approach has firm foundations for both deterministic chaotic dynamical systems \citep{Ruelle2009} and their stochastic counterpart  \citep{Hairer2010}. Indeed, linear response theory allows us to compute the system’s response to disturbances using response operators acting on the unperturbed system, hence our framework provides practical ways of designing ensemble experiments from which resilience through adaptation can be estimated. 

Important insights into the intricacies of the effect of disturbances to complex systems have been elucidated through linear response theory. For example, \citet{Held2004,Lenton2008} showed that the divergence of the response operator occurs in the proximity of a system to critical behavior due to bifurcations. Moreover the presence of very high sensitivity to perturbations has been showed to be connected to the presence of the so-called critical slowing down, {\it i.e.},  the presence of slow decay of correlations \citep[see][]{Scheffer:2009aa,santos2022}.  As shown in \cite{boettner22} both phenomena are due to the near-prevalence of positive, destabilizing feedbacks over the negative, stabilizing ones. In other words, adaptation can be slow and inefficient, which reduces the overall system's resilience. 

In the context of ecological resilience, to the best of our knowledge, this is the first study that applies linear response theory to perturbations on noisy forced systems. In fact, very few articles have considered ecological resilience for stochastic systems. The stochastic model of \citet{arnoldi2016resilience}, for example, treats noise as the source of perturbation to a linear deterministic system and quantifies resilience as the degree of stochastic variability around the equilibrium point. This type of analysis, which we call \textit{path resilience} is not the focus of the present study. For us, noise is a fundamental component of the dynamics, and we consider structural perturbations, namely changes to the model parameters. The work of \cite{ives1995measuring} shares some similarities with our work but is much more limited in scope. There, the author  investigates the response of populations interacting via a particular logistic model to perturbations on intrinsic growth rates. Resilience in \cite{ives1995measuring} is not measured with respect to adaptation, but as a change of the population's statistical variability per unit perturbation.  

This article focuses on the conceptual and mathematical basis of ecological resilience rather than on the intricacies of a given model. The formulation is quite general but inspired by applications to the quantification of resilience of living systems, in particular plants and, in principle, should be generalizable to complex systems comprising multiple biological processes. Throughout, we highlight along the mathematical derivations those concepts or processes of plant resilience that could be modeled by our framework, although application to real plant-system models is left for future work. For illustration, we use a simple bifurcating stochastic differential equation to exemplify our notion of resilience, as to not be encumbered by the complexities of a natural system.

The organization is as follows. In Section \ref{sec:model}, we describe the model that conceptualizes the dynamics of observables. These observables are described as time dependent multi-dimensional probabilistic distributions associated with the organism function. Section \ref{sec:resilience} introduces our measure of resilient adaptation along with a comparative analysis contrasting our proposal with existing methods. In Section \ref{sec:example} we carry out the process of quantifying the resilience to adaptation to a simple and familiar stochastic dynamics problem described by a Langevin equation. Finally, in Section \ref{sec:conc} we summarize the proposed measure of resilience to adaptation measure and discuss the assumptions and conditions required for its applicability

\section{Preliminaries}
\label{sec:model}

We know describe the mathematical framework upon which ecological resilience through adaptation is formulated. Our focus is on the temporal dynamics of a state variable $X(t) \in \R^N$ representing chemical and/or physical variables within a living system. In the case of plant physiology these could be the variables involved in a photosynthesis model ({\it e.g.}, water potential, carbohydrate concentrations or flows, chlorophyl, stomata aperture). We will also assume that there are identifiable environmental  stressors $\e \in \R^P$ ({\it e.g.}, solar radiation, ambient temperature, ambient or soil moisture) that have known or measurable effects on the dynamics of $X$. The variables that take on the role of adaptation variables also appear as  parameters in the dynamics of $X$. We label the adaptation variables by $\alpha \in \R^Q$. Our theoretical model for the evolution of the observable subset of the state variables is a stochastic differential equations \cite{pavliotis} of  the form 
\begin{equation}
 \ud X = F_{\e,\alpha}(X) \ud t + \sigma \ud W_t, \quad t>0,  \quad X(0) \sim p_0. 
 \label{eq:model}
 \end{equation}
 The drift term $F_{\e,\alpha}$ is known and of dimensions $\R^N$ for constant values of $\e$ and $\t$. The incremental Wiener processes $\ud W_t$ has the same dimensions as $F_{\e,\t}$. The assumed constant $\sigma \in \R^{N\times N}$ is a non-negative symmetric noise amplitude matrix. 

Model \eqref{eq:model} includes two sources of variability in the dynamics of $X$: for any fixed initial starting condition $X(0)$, an infinite amount of $X(t)$ histories can be generated (each corresponding to a realization of $W$), reflecting the random/noisy nature of the differential equation itself due to aleatoric or epistemic error. The other source of variability is encoded in the initial conditions: To model sub-system variability, we will assume that $X(0)$ is taken from known probability distribution $p_0 \in \R^N$. This aspect of the model accommodates for variability in the biological system itself ({\it e.g.}, variability in leaves of the same plant).

The function $F_{\e,\alpha}$ encodes all the known deterministic regulatory dynamics of the system. The noise represents random fluctuations that are present in the dynamics of $X$ and are modeled by an additive diffusion process which, in this work, are assumed independent of $\e$ and $\t$ although the analysis can be naturally extended. For current plant function models at leaf scale \citep[see][for example]{fatichi2014moving,fatichi2016modeling} the state variable has dimension $N \sim 10$ and include concentrations and fluxes important for carbon assimilation and transpiration, as well as vascular transport. The environmental variables $\epsilon$ are fewer and determine atmospheric and soil boundary conditions.  The distinction between the adaptation parameters in $\t$ and the state variables in $X$ is more subtle. Typically $\t$ includes regulatory or control variables whose dynamics are left out equation \eqref{eq:model} either because they operate at different length scales, are poorly understood, or can be deliberately changed in experiments. In the case of plants experiencing water deficit, for example, regulatory variables include the concentration abscisic acid, ions of calcium and potassium, and stomatal aperture \citep{willey2018environmental}.

In order to emphasize the dependence of the system on the parameters, we denote by $X_{\e,\t} = \{X_{\e,\t}(t): t\geq 0\}$ the solution to equation \eqref{eq:model} for constant values of $\e$ and $\t$. The key to our proposal is the mean evolution of observables of the process, namely expectations of the form $\EXP S(X_{\e,\t}(t))$ where $S: \R^N \to \R$ is some measure of success or well-being to be discussed further below. For now, in what follows, we explain the mathematical notation and background required for estimating the response of $\EXP S(X_{\e,\t}(t))$ to changes in $\e$ and $\t$, in terms of linear response theory. For details see \cite{pavliotis}. 

We suppose that that the drift $F_{\e,\t}(x)$ in \eqref{eq:model} is sufficiently smooth as a function of the state variable $x$ so that the strong solution to equation \eqref{eq:model} is a diffusion process. We also assume that $F_{\e,\t}$ is differentiable with respect to the parameters $\e$ and $\t$. For simplicity, we also suppose that the noise amplitude matrix is diagonal and isotropic $\sigma = \sigma I_N$ for some $\sigma >0$. The process $X_{\e,\t}$ has an infinitesimal generator
\begin{equation}\label{def:LopB}
    \Lop_{\e,\t}[S] = F_{\e,\t}\cdot \grad{S} + \frac{\sigma^2}{2} \grad^2 S, 
\end{equation}
for all functions $S$ in the domain $\Dom(\Lop_{\e,\t})$ which is supposed to be independent of $\e$ or $\t$, and dense within a Banach space $\mathcal{B}(\R^N)$. The operator $\Lop_{\e,\t}$ determines the evolution of the semigroup\\
\begin{equation}
    \EXP_x \, S(X_{\e,\t}(t)) = e^{t \Lop_{\e,\t}}[S](x), \quad t \geq 0
\end{equation}
for all $S \in \Dom(\Lop_{\e,\t})$. The subscript $x$ in the expectation denotes conditioning on the initial value $X_{\e,\t}(0) = x$. For $\lambda >0$, we denote by $\Rop^{(\lambda)}_{\e,\t}$ the resolvent operator of $\Lop_{\e,\t}$. Namely 
\begin{equation}\label{def:Resolvent}
    \Rop^{(\lambda)}_{\e,\t}[S] = (\lambda - \Lop_{\e,\t}[S])^{-1} = \int_0^\infty e^{-\lambda t} e^{t \Lop_{\e,\t}}[S](x)  \ud t.
\end{equation}
for any $S \in \Dom(\Lop_{\e,\t})$. The adjoint to \eqref{def:LopB} is the Fokker-Planck operator
\begin{equation}\label{def:LopF}
    \Lop^*_{\e,\t}[p] = -\div (p F_{\e,\t}) + \frac{\sigma^2}{2} \grad^2 p.
\end{equation}
For a initial distribution $p_0$, the evolution of the probability distribution $p_{\e,\t}(t,x)$ of $X_{\e,\t}(t)$ conditional to $X_{\e,\t}(0) \sim p_0$ evolves according to the `forward equation',
\begin{equation}
    \pderiv{p_{\e,\t}}{t} = \Lop^*_{\e,\t}[p_{\e,\t}], \quad p_{\e,\t}(0) = p_0,
\end{equation}
namely $p_{\e,\t}(t) = e^{t\Lop^*_{\e,\t} }[p_0]$. Hence expectations of an observable can be computed as  
\begin{equation}\label{eq:EpoS}
    \EXP_{p_0}\, S(X_{\e,\t}(t)) = \int_0^t S(x) e^{s\Lop^*_{\e,\t} }[p_0](x) \ud s
\end{equation}
for any bounded $S:\R^N \to \R$. Lastly, we will assume that for all $\e,\t$ of interest, the diffusion $X_{\e,\t}$ is ergodic with unique invariant probability $\bar{p}_{\e,\t}$. Namely $\Lop^*_{\e,\t}[\bar{p}_{\e,\t}] = 0$ and expectations can be computed as
\begin{equation}\label{eq:ergodicity}
        \lim_{t \to \infty} \EXP_x(S(X_{\e,\t}(t))) = \EXP_{\bar{p}_{\e,\t}}(S(X_{\e,\t}(t))) = \int S(x) \bar{p}_{\e,\t}(x) \ud x =: \bar{S}_{\e,\t}
\end{equation}
for all $t\geq 0$ and $x \in \R^N$, which amounts to assuming that for any ensemble described by $p_0$, the process $X_{\e,\t}(t)$ asymptotically converges in probability to the invariant measure $\bar{p}_{\varepsilon,\alpha}$.

\section{Resilient Adaptation}
\label{sec:resilience}

We will propose a measure of resilience aimed at quantifying the ability of the system to adapt through changes on $\t$, to environmental disturbances on $\e$ that are stressful with respect to some measure of performance, well-being, productivity, or success. 

\subsection{Background}
\label{subsec:background}

It is common to consider the resilience of deterministic homeostatic systems evolving along a stability landscape. Namely a model of the form $\ud X/\ud t = F_{\e,\t}(X)$ where $F_{\e,\t} = \grad V_{\e,\t}$ and $V_{\e,\t}$ is a potential surface. Resilience is studied by perturbing the state $X$ away from a stable equilibrium and analyzing its homeostatic relaxation to the same or other stable equilibria \citep[see][]{van2021unifying,Dakos22}. 

The noisy case presented in equations \eqref{eq:model} has been widely studied in non-equilibrium statistical mechanics literature as a model of a system evolving towards a potential energy minimum with random fluctuations (see \cite{brenig2012statistical,pavliotis}, for example). The paths of $X$ do not necessarily converge to equilibrium states, but will actually transition randomly between attractors due to the combination of noise or forcings. This stochastic homeostatic behavior has been used by \citet{Meyer16,Krakovska16,arani2021exit} and others to study what we call `path-wise resilience' to random perturbations. Namely, paths generated from any initial state drift towards the neighborhood of the local minima of $V_{\e,\t}$ just as in the deterministic case, but are continuously subjected to random fluctuations whose magnitude depend on the amplitude $\sigma$ of the noise process. If $\sigma$ is sufficiently large, or enough time passes, these fluctuations will drive paths across any unstable equilibrium into another basin of attraction where the system will recover in a different homeostatic state.  In fact, estimates for the probability of such changes and the average time they take are well-known for the dynamics given by equation \eqref{eq:model} and are described in the weak-noise limit by large deviation laws  \cite{freidlin1998,gardiner,arani2021exit}. This, however, is not our focus. We are interested in structural resilience with respect to parameter changes, and on the average over all possible paths. 

In contrast with the two approaches described above, we regard biological systems as \textit{homeodynamic}, and assume that they can transform their dynamics through behavioral changes in response to perturbations \citep{lloyd2001homeodynamics}. We refer to this behavioral changes as adaptations and model them as changes $\t \to \t + \Delta \t$ in response to environmental perturbations of the form $\e \to \e + \Delta \e$. The homeodynamic response, when described in terms of temporally-dependent distributions, can be analyzed using linear response theory. Specifically, we can obtain first order estimates to the sensitivity of the expectation of key observables associated with disturbances to changes $\Delta \e$ and $\Delta \t$. 

Further, we conceptualize ecological resilience in terms of a \textit{success function} $S: \R^N \to [0,\infty)]$; an observable that measures well-being, fitness, productivity, etc., as a function of the system state $X_{\e,\t}(t)$ at any given time, and under specific operating conditions determined by parameters \( \e \) and \( \t \). States $x$ for which $S(x)$ is close to zero are associated with biological stress. For example, in plants, \( S \) could represent photosynthetic output under various environmental conditions. Mathematically, we assume $S$ is a bounded function belonging to $\Dom(\Lop_{\e,\t})$ for all values of interest of $\e,\t$. 

The study focuses on the dynamics of the mean of \( S(X_{\e,\t}(t)) \) rather than the potentially complex and high-dimensional sample paths of \( X_{\e,\t}(t) \). Specifically, resilience is quantified in terms of  how the average success, expressed as \( \EXP S(X_{\e,\t}(t)) \), reacts to changes to $\e$ and $\t$. The expectation here is meant as an ensemble mean over a large group of individuals or subsystems.

\subsection{Perturbation, Adaptation and Resilience}

The solution $X_{\e,\t}$ to \eqref{eq:model} represents the dynamics of the system under constant environmental and metabolic conditions. In the context of resilience we are interested in the ensemble properties of the solution under changing values of $\e$ and adaptations of $\t$. Specifically, we consider a process $X$ defined by the following three steps: For $t<0$ the system is evolving in probabilistic equilibrium under the operation conditions $(\e_0,\t_0)$, namely $X(t)\sim \bar{p}_{\e_0,\t_0}$ for all $t<0$. At $t=0$ a sudden environmental disturbance $\e_0 \to \e_1$ occurs. After a random time $\tau \sim \exp(\lambda)$ the system adapts by switching $\t_0 \to \t_1$. See Figure \ref{fig:cartoon}. For $t \geq 0$ the resulting process can be described as
\begin{equation} \label{def:fullX}
X(0) \sim \bar{p}_{\e_0,\t_0}, \quad 
    X(t) = \begin{cases}
        X_{\e_1,\t_0}(t) & 0<t \leq \tau \\
        X_{\e_1,\t_1}(t) & t>\tau
    \end{cases}
\end{equation}
where the continuity condition $X_{\e_1,\t_0}(\tau^-) = X_{\e_1,\t_1}(\tau^+)$ is assumed to hold with probability one.

\begin{figure}
    \centering
    \includegraphics[scale=0.8]{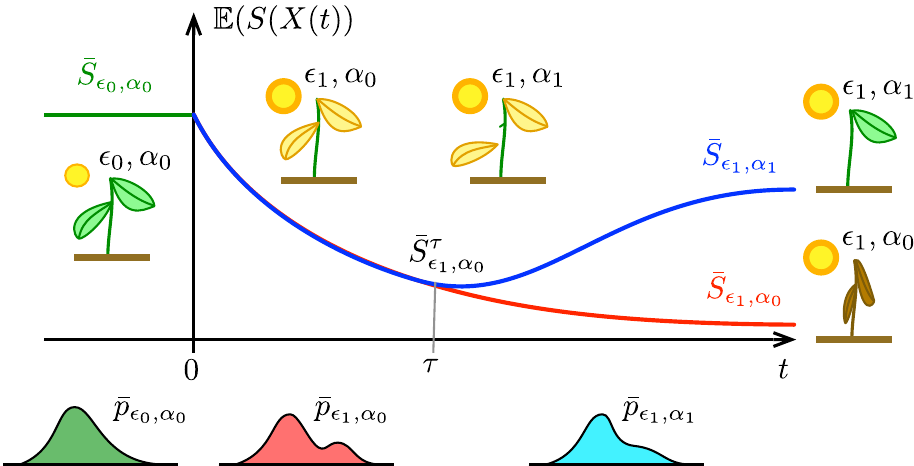}
    \caption{Schematic representation of the evolution of the expected success $\EXP S(X(t))$ for $X$ given by \eqref{def:fullX}. For $t<0$ the system is in statistical equilibrium under $\bar{p}_{\e_0,\t_0}$. At $t=0$ the perturbation $\e_0 \to \e_1$ occurs (e.g. temperature increases) and $S(X(t))$ starts decreasing in average. If no adaptation takes place, the red curve is followed towards the equilibrium distribution $\bar{p}_{\e_1,\t_0}$. Adaptation occurs at the random time $\tau$ and the resilient system recovers to the equilibrium distribution $\bar{p}_{\e_1,\t_1}$.}
    \label{fig:cartoon}
\end{figure}

Note that the adaptation $\t_0 \to \t_1$ is assumed to occur instantaneously at a random `reaction time' $\tau$ representing the time it takes for the regulatory signals to take effect. Randomness here means that each individual suffering the disturbance might react at a different time, but in the ensemble, these times follow the prescribed exponential distribution with mean $1/\lambda$ units of time. This modeling assumption can be interpreted as consistent with a case in which the underlying adaptation dynamics $\ud \t = g(\e,X) \ud t$ are poorly understood or deliberately left out of the model \eqref{eq:model}, or unresolved at the time scales of $t$. The choice of the exponential distribution for $\tau$ is parsimonious, and has the mathematical advantage of making $X$ a Markov process.

We are interested in the dynamical behavior of the ensemble mean $\EXP(S(X(t))$ of the success throughout the sequence of invariance-perturbation-adaptation. For constant $\e,\t$ we denote by $\bar{S}_{\e,\t}$ the mean of $S$ under the invariant distribution $\bar{p}_{\e,\t}$. Note that by the construction of $X$, 
\begin{equation}\label{eq:ES0inf}
    \EXP(S(X(0))  = \bar{S}_{\e_0,\t_0}, \quad \lim_{t \to \infty}  \EXP(S(X(t)) = \bar{S}_{\e_1,\t_1}.
\end{equation}
In the context of resilience we study the trajectory in time of the ensemble average of the success between the extremes in \eqref{eq:ES0inf}. See Figure \ref{fig:cartoon}.

Concepts related to resilience to parameter disturbances are usually defined with respect to the initial, worst, and final states of an observable during the process of perturbation/response/recovery. See for example \cite{yi2021review,van2021unifying}. In the context of the time evolution of the expected success, the usual framework proposes the following definition of resilience: 
\begin{equation}\label{def:usualR}
    R = \frac{\bar{S}_{\e_1,\t_1}- \min_{t \geq 0} \EXP(S(X(t))}{\bar{S}_{\e_0,\t_0} - \min_{t \geq 0} \EXP(S(X(t))}
\end{equation}
which in \cite{yi2021review} is called `recovery', and equals one minus `recovery capacity' over `resistance' in the notation of \cite{van2021unifying}. 

One drawback of definition \ref{def:usualR} is that the term $\min_{t \geq 0} \EXP(S(X(t))$ is not an ensemble average. Furthermore its computation requires estimation of $\EXP S(X(t))$ all times. To address this, we exploit the assumption of the existence of a population-wide reaction time with known probability distribution and define `resilience adaptation' as follows.

\begin{definition}\label{def:Rtau}
  Let $X$ be the solution to \eqref{def:fullX}, $S$ the success function and $\tau$ the adaptation time. Denote $\bar{S}^\tau_{\e_1,\t_0} := \EXP S(X(\tau))$. Then the resilience of $X$ to the perturbation $\e_0 \to \e_1$ and through the adaptation $\t_0 \to \t_1$ is   
  \begin{equation}\label{eq:Rtau}
    R_\tau = \frac{\bar{S}_{\e_1,\t_1}- \bar{S}^\tau_{\e_1,\t_0} }{\bar{S}_{\e_0,\t_0} - \bar{S}^\tau_{\e_1,\t_0}}.
\end{equation}
\end{definition}

We will argue that $R_\tau$ is a practical and informative measure of ecological resilience, and that is consistent with existing methodologies. 

\subsection{Computing and interpreting Resilience}
\label{sec:Practical}

The resilience $R_\tau$ in \eqref{eq:Rtau} can be any real number, and has the following interpretation:
\begin{itemize}
    \item In the case of stressful environmental perturbations, where $\bar{S}_{\e_0,\t_0} > \bar{S}^{\tau}_{\e_1,\t_0}$, $R_\tau$ is negative only if the adaptation $\t_0 \to \t_1$ is also detrimental with respect to the mean state of the system when the adaptation occurs. 
    \item A value $0<R_\tau<1$ measures the fraction of the success that the system was able to recover with the adaptation, with respect to the total loss of success due to the disturbance. 
    \item A value $R>1$ indicates a very resilient system, in which the long-term success after the adaptation $\bar{S}_{\e_1,\t_1}$ is larger than the initial $\bar{S}_{\e_0,\t_0}$.
\end{itemize}

With respect to the usual definition or resilience $R$ in \eqref{def:usualR}, we  note that is not equal to $R_\tau$ in general, although the example in Section \ref{sec:example} shows that it can be a good approximation. We argue, however that $R_\tau$ is a more practical measure of ecological resilience than $R$ because can be computed in terms only of averages. The main argument stems from the fact that while $\min \EXP S(X(t))$ is a deterministic ensemble diagnostic, the reaction time
$\tau$ pertains to the physiological ability of the system to adapt or heal, and is hence a random variable that can be modeled. This, as described below, opens up several possibilities for computation of $R_\tau$.

The usual resilience $R$ can be, in principle, computed from experiments that involve the comprehensive monitoring of a large population. In such an experiment, up to Gaussian errors, each individual must be subjected to the same environmental disturbance $\e_0 \to \e_1$ and react with the (possibly unknown) adaptation $\t_0 \to \t_1$ that has measurable effects on the success throughout time. The average of $S$ must be computed at enough times to discern the minimum value of $\EXP S(X(t))$. In this experimental context, computing $R_\tau$ might not be very practical because estimating $\bar{S}^\tau_{\e_1,\t_0}$ would require measuring the stress of each individual at the exact moment it adapts. However, if a model for the dependence of $X$ on $\e$ and $\t$ as in \eqref{def:fullX} is available, then $R_\tau$ is a more practical measure than $R$ because it can be estimated from an ensemble of simulations of fixed dynamics under fixed initial distributions. 

Furthermore, note that by \eqref{def:fullX}, we can write $\bar{S}^\tau_{\e_1,\t_0}$ by taking expectations jointly over $\tau$ and the paths of $X_{\e_1,\t_0}$ conditioned on $X_{\e_1,\t_0}(0) \sim \bar{p}_{\e_0,\t_0}$. This yields the following expression in terms of the resolvent (see equation \eqref{def:Resolvent}), 
\begin{equation}\label{eq:S10tau}
    \bar{S}^\tau_{\e_1,\t_0} =\EXP_{\bar{p}_{\e_0,\t_0}}\left[S(X_{\e_1,\t_0}(\tau)) \right] 
    = \lambda \int \Rop_{\e_1,\t_0}^{(\lambda)}[S](x) \, \bar{p}_{\e_0,\t_0}(x) \ud x.
\end{equation}
Hence, if a model $F_{\e,\t}$ is at hand and the operator $\lambda - \Lop_{\e,\t}[S]$ can be analytically or numerically inverted, one can estimate $\bar{S}^\tau_{\e_1,\t_0}$ only from observations or simulations of the unperturbed system. This feature is key, as we will demonstrate in Section \ref{sec:Estimation}, it yields useful estimation methods for $R_{\tau}$.

Our proposal $R_\tau$ can be related to other resilience measures applicable to noisy systems. First, in \citet{ives1995measuring}, resilience is measured with respect to perturbations that increase the variability of a dynamically evolving population. Namely, the variance of $X(t)$ is used as a an inverse metric of success. Although the analysis in \citet{ives1995measuring} is limited to rate of increase of the variance after an environmental perturbation, we can write Definition \eqref{def:Rtau} with 
respect to the variance as
\begin{equation}\label{def:Rvar}
    R^{\VAR}_\tau = \frac{\VAR (X(\tau)) - \VAR_{\bar{p}_{\e_1,\t_1}} (X) }{\VAR (X(\tau)) - \VAR_{\bar{p}_{\e_0,\t_0}} (X)}
\end{equation}
where $\VAR_{\bar{p}_{\e,\t}} (X)$ simply denotes the variance of the probability distribution defined by $\bar{p}_{\e,\t}$ and  $\VAR (X(\tau))$ is the variance of the process \eqref{def:fullX} at time $\tau$. Note that, with respect to \eqref{def:Rtau}, the signs of the numerator and denominator of \eqref{def:Rvar} where reversed, which has no effect except for emphasizing that both are positive quantities. The example in Section \ref{sec:example} illustrates the correspondence between $R_\tau$ and $R_{\tau}^{\VAR}$. 

Another quantity used to measure resilience is that of the characteristic return time to equilibrium after perturbation \citep[see][]{ives1995measuring,Meyer16,Krakovska16}. In the context of our formulation, this concept is related to the rate at which $\EXP S(X(t))$ diverges from $\bar{S}_{\e_0,\t_0}$ towards $\bar{S}_{\e_1,\t_0}$ right after the environmental perturbation, and the rate at which $\EXP S(X(t))$ converges to $\bar{S}_{\e_1,\t_1}$ after the adaptation. These rates are encoded in the largest non-zero eigenvalue $\rho_{\e,\t}$ of the forward operator $\Lop^*_{\e,\t}$ in \eqref{def:LopB} at the different stages of the process. The proposed resilience measure $R_\tau$ can be viewed as a comparison between the values of the inverse time-scales $\rho_{\e_1,\t_0},\lambda$ and $\rho_{\e_1,\t_1}$. A resilient system would be one in which $|\rho_{\e_1,\t_1}|$ is large compared to $|\rho_{\e_1,\t_0}|$.

\subsection{Estimating Ecological Resilience}
\label{sec:Estimation}

Linear response theory can be used to estimate the effect that a perturbation has on the distribution and averages of a non-linear stochastic process. Namely, assuming both $\Delta \e := \e_1-\e_0$ and $\Delta \t := \t_1 - \t_0$ are small, we will use linear response theory to give estimates to $\bar{S}_{\e_1,\t_1} - \bar{S}_{\e_0,\t_0}$ and $\bar{S}^{\tau}_{\e_1,\t_0} - \bar{S}_{\e_0,\t_0}$ which can, in turn, be used to give approximations to $R_\tau$ in \eqref{def:Rtau}.\\

For definiteness, consider the process $X_{\e_1,\t_0}$ with initial distribution $X_{\e_1,\t_0}(0) \sim \bar{p}_{\e_0,\t_0}$. By choosing different values for $\e$ in the dynamics for $t>0$ and the initial distribution, we are modeling a system that experiences a press disturbance $\e_0 \to \e_1$ for all $t > 0$ (see Figure \ref{fig:cartoon}). Since we are assuming that $F_{\e_0,\t_0}$ in \eqref{eq:model} is differentiable with respect to $\e$, the disturbance produces a perturbation on the drift, which to first order on $\Delta \e$ is
\begin{equation}
    F_{\e_0,\t_0} \to F_{\e_0,\t_0} + \jac_\e F_{\e_0,\t_0} \Delta \e
\end{equation}
where $\jac_\e F_{\e_0,\t_0}$ denotes the Jacobian of the vector field $F_{\e_0,\t_0}$  with respect to the vector parameter $\e \in \R^P$. Linear response provides an approximate expression for the probability density of the disturbed process and averages of any observable, in terms of the un-disturbed distribution $\bar{p}_{\e_0,\t_0}$. Note that the corresponding perturbation on the Fokker-Planck operator is
\begin{equation}\label{def:ell}
    \Lop^*_{\e_0,\t_0}[p] \to \Lop^*_{\e_0,\t_0}[p] - \div(p \,\jac_\e F_{\e_0,\t_0} \, \Delta \e)  
    = \Lop^*_{\e_0,\t_0}[p] + \Delta \e  \cdot \ell_\e[p],
\end{equation}
where the operator $\ell_\e$ is defined coordinate-wise as follows
\begin{equation}\label{def:ell_ij}
    \ell_\e[p]^{(i)} = - \sum_{j=1}^N \pderiv{}{x^{(j)}} \left(p \pderiv{F^{(j)}_{\e_0,\t_0}}{\e^{(i)}}\right), \quad i =1,\dots,P.
\end{equation}

If we approximate  to first order the expectation of $S$ under the perturbation as,
\begin{equation}\label{eq:ES0app}
        \EXP_{\bar{p}_{\e_0,\t_0}}(S(X_{\e_1,\t_0}(t)) \approx \bar{S}_{\e_0,\t_0} + 
   \Delta \e \cdot \Delta_\e \bar{S}_{\e_0,\t_0}(t)
\end{equation}
then linear response theory says that the correction $\Delta_\e \bar{S}_{\e_0,\t_0}(t)$ can be written in terms of the forward evolution operator as
\begin{equation}
    \Delta_\e \bar{S}_{\e_0,\t_0}(t) = 
   \int_0^t \int e^{s \Lop_{\t_0,\e_0}^*}\left[ \ell_\e[\bar{p}_{\e_0,\t_0}]\right](x) \, S(x)  \ud x \ud s. \label{eq:LRT1}
\end{equation}
See \cite{pavliotis} for details. 

Note that \eqref{eq:LRT1} is written in terms exclusively of the un-perturbed dynamics and, by \eqref{eq:EpoS}, can be expanded out as a correlation over paths of the process $X_{\e_0,\t_0}$,
\begin{equation}
    \Delta_\e \bar{S}_{\e_0,\t_0}(t) = \int_0^t \EXP_{\bar{p}_{\e_0,\t_0}}\left\{ r_\e[\bar{p}_{\e_0,\t_0}] (X_{\e_0,\t_0}(0)) S(X_{\e_0,\t_0}(s)) \right\} \ud s \label{eq:LRTE1} 
\end{equation}
where $r_\e$ denotes the operator
\begin{equation}
    r_\e[p] = \frac{\ell_\e [p]}{p}
\end{equation}
for suitable $p: \R^N \to \R$.

The approximation \eqref{eq:ES0app} at $t=\tau$ provides an estimate for 
\begin{equation}\label{eq:DEStau}
    \bar{S}^\tau_{\e_1,\t_0} \approx  \bar{S}_{\e_0,\t_0} + \Delta \e \cdot \Delta_\e \bar{S}^\tau_{\e_0,\t_0}.
\end{equation}
To estimate the perturbation $\Delta_\e \bar{S}^\tau_{\e_0,\t_0}$ we can multiply \eqref{eq:ES0app} times $\lambda e^{-\lambda t}$ and integrate with respect to $t$. We obtain,
\begin{align}
    \Delta_\e \bar{S}^\tau_{\e_0,\t_0} &= 
    \int \lambda \int_0^\infty  e^{-\lambda t}\int_0^t  e^{s \Lop_{\t_0,\e_0}}[S](x) \ud s \ud t \,  \ell_\e[\bar{p}_{\e_0,\t_0}](x) \ud x\\
    &=\int \int_0^\infty e^{-\lambda s} e^{s \Lop_{\t_0,\e_0}}[S](x) \ud s  \,  \ell_\e[\bar{p}_{\e_0,\t_0}](x) \ud x\\
    &= \int \Rop^{(\lambda)}_{\e_0,\t_0}[S](x) \ell_\e[\bar{p}_{\e_0,\t_0}](x) \ud x.  \label{eq:LRTtau}
\end{align}
As previously mentioned, the estimate \eqref{eq:LRTtau} is useful whenever the resolvent can be computed analically or numerically. In general, we can also take expectations of \eqref{eq:LRTE1} with respect to $\tau$ to obtain an estimate of $\bar{S}^\tau_{\e_1,\t_0}$ as a correlation suitable for simulations,
\begin{equation}\label{eq:LRTtauE}
    \Delta_\e \bar{S}^\tau_{\e_0,\t_0} = \EXP_{\bar{p}_{\e_0,\t_0}}\left\{ r_\e[\bar{p}_{\e_0,\t_0}] (X_{\e_0,\t_0}(0)) S(X_{\e_0,\t_0}(\tau)) \right\}.
\end{equation}


Replacing the terms $\bar{S}^\tau_{\e_1,\t_0}$ in the expression of \eqref{def:Rtau} by its approximation in \eqref{eq:DEStau} we obtain the following more practical expression for the resilience,
\begin{equation}\label{eq:appRtau}
    R_\tau \approx \tilde{R}_\tau = 1 - \frac{  \bar{S}_{\e_0,\t_0}- \bar{S}_{\e_1,\t_1}}{-\Delta \e \cdot \Delta_\e \bar{S}^\tau_{\e_0,\t_0}}.
\end{equation}
Note that the numerator in \eqref{eq:appRtau} contains the extremes values in \eqref{eq:ES0inf}, namely the `before and after' of a population that has gone through the invariance-perturbation-adaptation sequence. The denominator is positive for stressful environmental perturbations and models the average `damage' caused by the perturbation. It provides an alternative to calculating the challenging term $\bar{S}^\tau_{\e_1,\t_0}$ in \eqref{def:Rtau}. The term that is subtracted from one is therefore the total long-term change in average success as a fraction of the total damage. 


A further estimate for the numerator $\bar{S}_{\e_1,\t_1}- \bar{S}_{\e_0,\t_0}$ in \eqref{eq:appRtau} may be obtained for small $\Delta \e$, $\Delta \t$ by straightforward differentiation. Assuming  smoothness of the invariant distribution and of $S$ with respect to $\e$ and $t$, we can write to first order
\begin{equation}\label{eq:appDES_grad}
    \bar{S}_{\e_1,\t_1}- \bar{S}_{\e_0,\t_0} \approx 
    \Delta \e \cdot \grad_\e \bar{S}_{\e_0,\t_0} + \Delta \t \cdot \grad_\t \bar{S}_{\e_0,\t_0}
\end{equation}
where $\grad_\e$ and $\grad_\t$ denote respectively the gradients with respect to the parameters $\e$ and $\t$ of the invariant expectation of the success at $\t = \t_0, \e = \e_0$. Equation \eqref{eq:appDES_grad} approximates the numerator in \eqref{eq:appRtau} as a combination of the sensitivities of the system to the different parameters, and yields yet another approximate expression for the resilience
\begin{equation}\label{eq:RtauEnAd}
    R_\tau \approx \hat{R}_\tau = 1 - \frac{\Delta \e \cdot \grad_\e \bar{S}_{\e_0,\t_0}}{\Delta \e \cdot \Delta_\e \bar{S}^\tau_{\e_0,\t_0}} 
    - \frac{\Delta \t \cdot \grad_\t \bar{S}_{\e_0,\t_0}}{\Delta \e \cdot \Delta_\e \bar{S}^\tau_{\e_0,\t_0}}.
\end{equation}

For the sake of argument, suppose the environmental parameter is one-dimensional ($P=1$) in \eqref{eq:RtauEnAd}, then
\begin{align}\label{eq:RtauEnAd1D}
    \hat{R}_\tau &= 1 - \frac{\grad_\e \bar{S}_{\e_0,\t_0}}{ \Delta_\e \bar{S}^\tau_{\e_0,\t_0}} 
    + \frac{\grad_\t \bar{S}_{\e_0,\t_0}}{-\Delta_\e \bar{S}^\tau_{\e_0,\t_0}} \cdot \frac{\Delta \t}{\Delta \e}\\
    &=: 1- R_{\tau}^{\rm Env} + R_{\tau}^{\rm Ad} \cdot \frac{\Delta \t}{\Delta \e} \label{eq:RtauParts}
\end{align}
The approximation $\hat{R}_\tau$ to $R^\tau$ separates the proposed resilience measured into its competing terms. The term $R_{\tau}^{\rm Env}$  is positive and depends only on the initial state, the effect of the environmental disturbance, and how much time the system takes to react. It re-scales the sensitivity of the system to changes in $\e$ with respect to the total damage. The term $R_{\tau}^{\rm Ad} \cdot \frac{\Delta \t}{\Delta \e}$ is positive for resilient systems. Its factor $\Delta \t/\Delta \e \approx \deriv{\t}{\e}(\e_0)$ represents the adaptation strategy (recall the discussion after \eqref{def:fullX} around the underlying unresolved adaptation dynamics) and $R_{\tau}^{\rm Ad}$ models the effect of the adaptation as compared to the damage incurred by the environmental perturbation.

Computing the gradients with respect to $\t$ and $\e$ in  expressions \eqref{eq:RtauEnAd} or \eqref{eq:RtauEnAd1D} requires evaluation or measurement of the dynamics under the invariant distribution of a perturbed system along each of the $P$ and $Q$ components of the parameters. For an alternate expression we can use the following result that, as in the linear response theory derivations, allows for computation of derivatives in terms only of expected values of a particular observable of the un-perturbed system. 
\begin{proposition}
If $S$ is independent of $\e$, then
\begin{equation}
   \grad_\e \bar{S}_{\e_0,\t_0} = \int S(x) \grad_\e \bar{p}_{\e_0,\t_0}(x) \ud x 
    = \EXP_{\bar{p}_{\e_0,\t_0}} (\jac_\e F_{\e_0,\t_0} \grad \psi_S)
\end{equation}\label{prop:gradS}
where $\psi_S$ is the solution to Poisson equation $-\Lop_{\e_0,\t_0} [\psi_S] = S$. Similarly for $\grad_\t \bar{S}_{\e_0,\t_0}$. 
\end{proposition}
\begin{proof}
    By the definition of $\psi_S$ we can write
    \begin{align*}
    \grad_\e \bar{S}_{\e_0,\t_0} &= -\int \Lop_{\e_0,\t_0} [\psi_S] \grad_\e \bar{p}_{\e_0,\t_0}(x) \ud x \\
    &= -\int  \psi_S(x) \Lop_{\t_0,\e_0}^*[ \grad_\e \bar{p}_{\e_0,\t_0}](x) \ud x.
    \end{align*}
Using the particular form of the forward operator \eqref{def:LopF} and the fact that $\Lop_{\e_0,\t_0}^*[\bar{p}_{\e_0,\t_0}]=0$, one can write
\begin{equation*}
    \Lop_{\t_0,\e_0}^*[ \grad_\e \bar{p}_{\e_0,\t_0}] = \grad \left(\bar{p}_{\e_0,\t_0} \jac_\e F_{\e_0,\t_0}\right).
\end{equation*}
Finally, integration by parts yields
\begin{align*}
    \grad_\e \bar{S}_{\e_0,\t_0} &= \int \grad \psi_S(x) \bar{p}_{\e_0,\t_0}(x) \jac_\e F_{\e_0,\t_0}(x) \ud x
\end{align*}
as desired.
\end{proof}

Note that a slightly more convoluted proof of Proposition \ref{prop:gradS} can be obtained using the linear response representation \eqref{eq:LRT1}, the identity $ \int_0^t  e^{s \Lop_{\t_0,\e_0}} \ud s = (I- e^{t \Lop_{\t_0,\e_0}})[(-\Lop_{\t_0,\e_0})^{-1}]$, and passing to the limit as $t \to \infty$. See \citet[section 9.3]{pavliotis}. A closely related formula for the sensitivity of a finite state Markov chains to general perturbations has been presented in \cite{Lucarini2016}.

\section{Example: a Gradient-Driven Stochastic Differential Equation}
\label{sec:example}
 We illustrate the measure of resilience through adaptation in  a  case  where the invariant distribution generated by  \eqref{eq:model} can be computed analytically. We consider a gradient-driven Langevin equation with a one-dimensional quartic potential (also known as Smoluchowski Diffusion Equation), 
\begin{align}
 \ud X &= -\pderiv{V_{\e,\t}}{x} (X) dt + \sigma \ud W_t, \quad t>0
 \label{eq:gradientexample}\\
 V_{\e,\t}(x) &= \t^{(1)} x^4 - \t^{(2)} x^3 + \e x - c, \quad x \in \R.
\end{align}
The constant $c$ ensures that $V_{\e,\t}(x)\geq 0$ for all $x \in \R$. The adaptation parameter is two-dimensional $\t = (\t^{(1)},\t^{(2)}) \in [0.5,1]\times[1,2]$. The environmental parameter takes values in $\e \in [0,2]$. Namely $P=1$ and $Q=2$, which exemplifies the typical case where  regulatory parameters are more numerous than environmental parameters.

For any combination $(\e,\t)$ in those ranges, the potential $V_{\e,\t}$ is confining and therefore the solution process $X_{\e,\t}$ to \eqref{eq:gradientexample} is ergodic with a unique invariant distribution density given by
\begin{equation}\label{def:pbarPot}
    \bar{p}_{\e,\t}(x) = \frac{1}{Z_{\e,\t}} \exp\left(-\frac{2}{\sigma^2} V_{\e,\t}(x)\right), \quad x \in \R,
\end{equation}
where $Z_{\e,\t}$ is a constant ensuring that $\bar{p}_{\e,\t}$ integrates to unity. The equilibrium states of the potential $V_{\e,\t}$ correspond to the roots of $\pderiv{}{x}V_{\e,\t}$ and to the local maxima of $\bar{p}_{\e,\t}$. In fact, the system undergoes a `supercritical pitchfork' bifurcation at
\begin{equation}\label{def:ebifur}
    \e_b(\t) = \frac{(\alpha^{(2)})^3}{2 (\alpha^{(1)})^2}
\end{equation}
having one stable equilibrium if $\e>\e_b$ and two if $\e<\e_b$ \citep{strogatz2018nonlinear}. See Figure \ref{fig:example}a.

\begin{figure}
    \centering
    \subfigure[]{\includegraphics[scale=0.7]{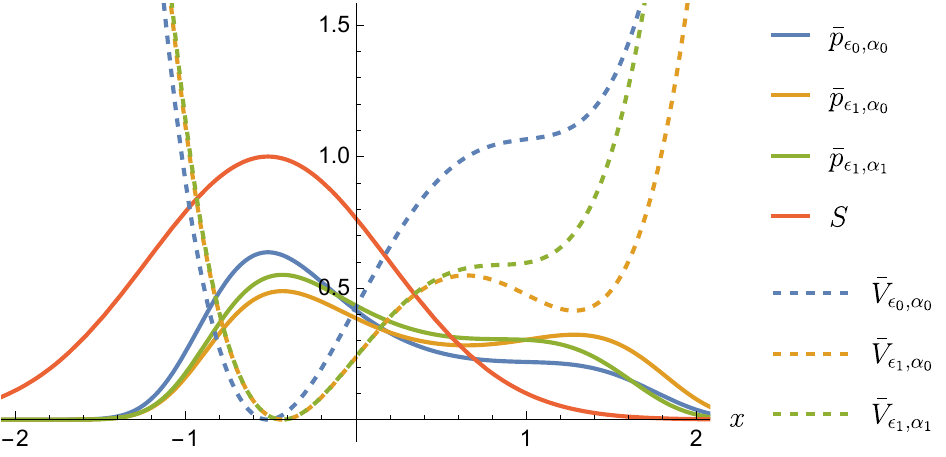}} 
    \subfigure[]{\includegraphics[scale=0.7]{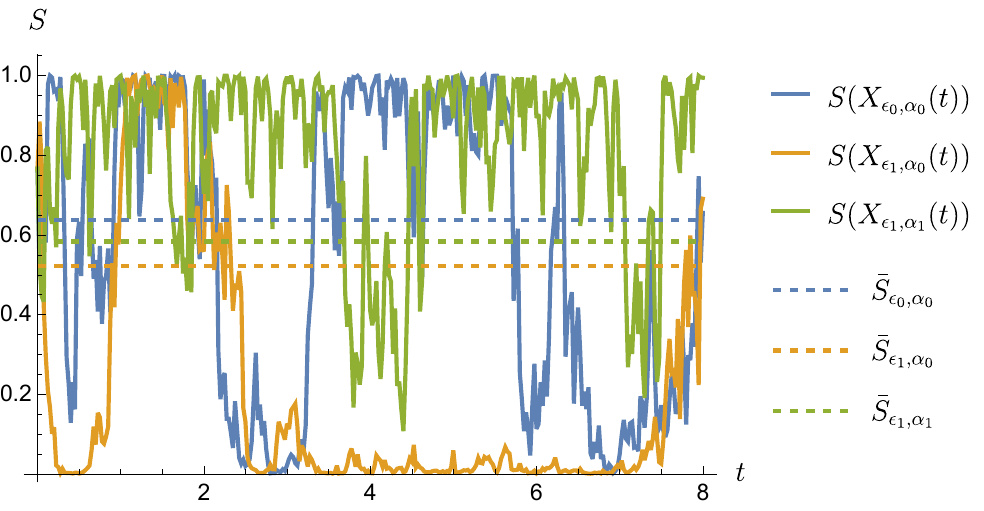}} 
    \caption{(a) Stationary distributions $\bar{p}_{\e,\t}$ and potentials $V_{\e,\t}$ for an invariance-perturbation-adaptation sequence with $\sigma^2=2$, $\e_0=1.2$ $\e_1 = 0.8$,  $\t_0=(0.56,1.12)$ and $\t_1 = (0.63,1.07)$. Note the unimodality for $(\e_0,\t_0) $ and $(\e_1,\t_0)$, and bimodality for $(\e_1,\t_1)$. The success function $S$ is shown in red for comparison. (b) Paths of $S(X_{\e,\t}(t))$ for one realization of the process \eqref{eq:gradientexample} in each of the parameters combinations in (a). Dashed horizontal lines show the mean of $S$ with respect to the corresponding invariant distribution. Note that the process spends much more time in low success states for the bimodal configuration.} 
    \label{fig:example}
\end{figure}

The disturbances to $\e$ and $\t$ in system \eqref{eq:gradientexample} are performed according to a protocol of invariance-perturbation-adaptation we now explain. Suppose that `normal' operating conditions encoded in $\e_0,\t_0$ satisfy $\e_0 > \e_b(\t_0)$ so that there is only one equilibrium solution which, for the range of parameters considered here, is near $x=-1/2$ (see Figure \ref{fig:example}a). We will consider this as the preferred/most successful state of the system. The environmental perturbation $\e_0 \to \e_1$  is such that  $\e_1 < \e_b(\t_0)$, inducing a pitchfork bifurcation that creates a second stable equilibrium in the positive real line. Paths of $X_{\e_1,\t_0}$ will then likely spend time in the basin of attraction of this second stable equilibrium, which we will presume is highly undesirable. The preference of $x=-1/2$ over the second equilibrium is is encoded for this example in the success function by defining
\begin{equation}\label{def:Sexample}
    S(x) = e^{-(x+1/2)^2}, \quad x \in \R.
\end{equation}
The value of $S(X(t))$ gives therefore an idea of how close the system is to `preferred' operating conditions.

The adaptation $\t_0 \to \t_1$ takes place instantly at a random reaction time $\tau \sim \exp(\lambda)$ after the environmental perturbation.  We consider the following adaptation strategy: $\Delta \t$ is a vector in the direction of $\grad_\t \bar{S}_{\e_1,\t_0}$ such that for $\t_1 = \t_0 + \Delta \t$, $\e_b(\t_1) < \e_1$ holds. The resulting adaptation is, by construction, such that the system reverts to the case $\e_1 > \e_b(\t_1)$ where only the preferred stable equilibrium remains and $\bar{p}_{\e_1,\t_1}$ is again unimodal. 

The choice of $S$ in \eqref{def:Sexample} and adaptation strategy ensures that the prescribed environmental and adaptation perturbations are, respectively, stressful and beneficial in average. Namely, in the notation of \eqref{eq:ergodicity}
\begin{equation}\label{eq:SbIneq}
    \bar{S}_{\e_0,\t_0} > \bar{S}_{\e_1,\t_0}, \quad  \bar{S}_{\e_1,\t_1} > \bar{S}_{\e_1,\t_0}.
\end{equation}
This strategy models the case in which the system can correct the bifurcation while at the same time ensuring a future mean success better than the do-nothing scenario $\bar{S}_{\e_1,\t_0}$. 

Figure \ref{fig:example}(a) depicts the potentials and stationary distributions in one example of the whole invariance-perturbation-adaptation sequence. Superimposed is the success function $S$. Figure \ref{fig:example}(b) shows the success function $S(X_{\e,\t}(t))$ along a single path for each regime. It highlights the noisy nature of the dynamics, including the episodic switches between basins of attraction in each of the three operating regimes, namely realizations exhibiting what in Section \ref{subsec:background} we referred to as `path-wise resilience'. The expectations of the success function with respect to each of the stationary distributions are depicted as well. 

Figure \ref{fig:exampleMeans} depicts the time evolution of the ensemble mean $\EXP S(X(t))$ for the process constructed as in \eqref{def:fullX} with the dynamics of \eqref{eq:gradientexample}, for a single combination of the invariance-perturbation-adaptation protocol $\e_0,\e_1,\t_0,\t_1$ described above. All the terms in the usual definitions of resilience, including $R$ given in equation \eqref{def:usualR}, can be `read-off' from this figure. See for example the discussions around Figure 4d in \cite{van2021unifying} and Figure 5 in \cite{yi2021review}. Note that $\EXP S(X(t))$ in Figure \ref{fig:exampleMeans} above varies between the initial and final values given by \eqref{eq:ES0inf} and that its minimum $\min_{t \geq 0} \EXP(S(X(t))$ is close to the value of $\bar{S}^\tau_{\e_1,\t_0}$.

\begin{figure}
    \centering
    \includegraphics[scale=0.7]{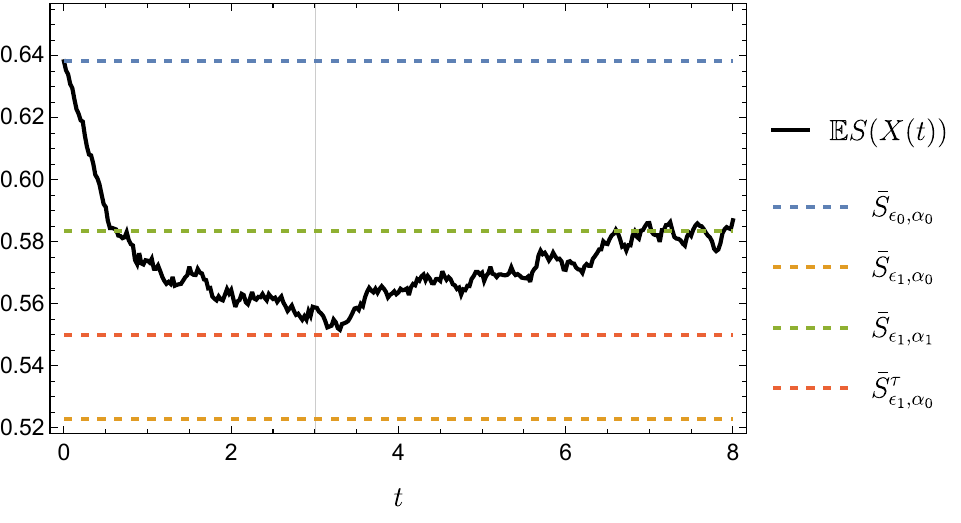}
    \caption{Ensemble mean $\EXP S(X(t))$ for $X$ given by \eqref{def:fullX} with $\sigma,\e_0,\e_1,\t_0,\t_1$ as in Figure \ref{fig:example}, and $\lambda = 1/3$. The mean was computed from a sample of 6000 paths. Horizontal dashed lines mark the invariant means of $S$ under each scenario and the mean of $S(X(\tau))$. The vertical line marks $t = \EXP(\tau)=1/\lambda$ The resilience is $R=0.44$. The value $\bar{S}_{\e_1,\t_0}$ is the long-term mean of the success in the do-nothing scenario, and plays no role in the computation of the resilience.}
    \label{fig:exampleMeans}
\end{figure}

\begin{figure}
    \centering
    \includegraphics[scale=0.75]{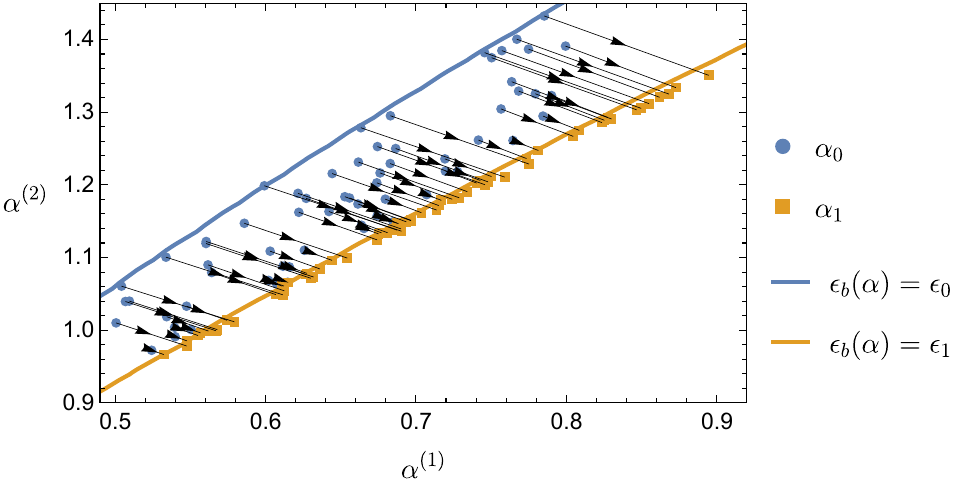}
    \caption{Values of $\t_0, \t_1$ used in the examples for numerically computation of $R$ and $R_\tau$. The solid lines are contour lines for the bifurcation threshold $\e_b$ in \eqref{def:ebifur} for $\e_0$ and $\e_1$ as in Figure \ref{fig:example}. Black arrows show that each adaption follows the strategy leading to equation \eqref{eq:SbIneq}: from $\e_b(\t_0)>\e_1$ and hence having two stable equilibria, to  $\e_b(\t_1) < \e_1$ and a single stable equilibrium after adapting.}
    \label{fig:disturbances}
\end{figure}

In order to robustly test our proposal $R_\tau$ in equation \eqref{def:Rtau} as a coherent measure of adaptation resilience, we conducted numerical experiments on different configurations of problem \ref{eq:gradientexample}.  We considered 64 different populations, each with a different value of $\t_0$. All populations undergo the same environmental bifurcation-inducing disturbance $\e_0 = 0.8 \to \e_1 = 1.2$. The adaptation strategy followed by the populations is as described above, and results in different a value of $\t_1$ for each example. Figure \ref{fig:disturbances} shows the location of $\t_0$ and $\t_1$ for each system with respect to bifurcation thresholds.

The numerical results are depicted in  Figure \ref{fig:RvRtau}. The dashed line indicates a perfect match between $R$ and $R_\tau$. We note that the data  is increasing, from which we conclude that $R_\tau$ can be used to establish when one system is more resilient than other in the usual sense. Note also that the two measures tend to agree for the most resilient examples, which in this case simply transpires the good approximation of $\bar{S}^{\tau}_{\e_1,\t_0}$ to $\min \EXP S(X(t))$. Those examples for which $R_\tau < 0$ correspond to cases in which $\bar{S}_{\e_1,\t_1} < \bar{S}^{\tau}_{\e_1,\t_0}$, namely the adaptation $\t_0 \to \t_1$ did not increase the mean success as compared with the mean success at $\tau$.

\begin{figure}
    \centering
    \includegraphics[scale=0.75]{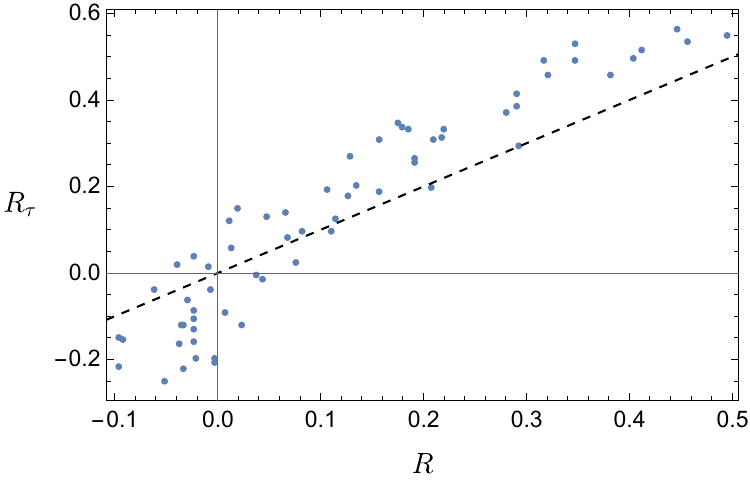}
    \caption{Comparison between $R$ and $R_\tau$ for the 64 combination of parameters shown in Figure \ref{fig:disturbances}. The dashed line marks equality.}
    \label{fig:RvRtau}
\end{figure}

Figure \ref{fig:RtauParts} shows the approximate components of $R_{\tau}$ as per equation \eqref{eq:RtauParts}. The environmental component is essentially constant throughout our examples, and most of the variation of the resilience is due to the adaptation component.
\begin{figure}
    \centering
    \includegraphics[scale=0.7]{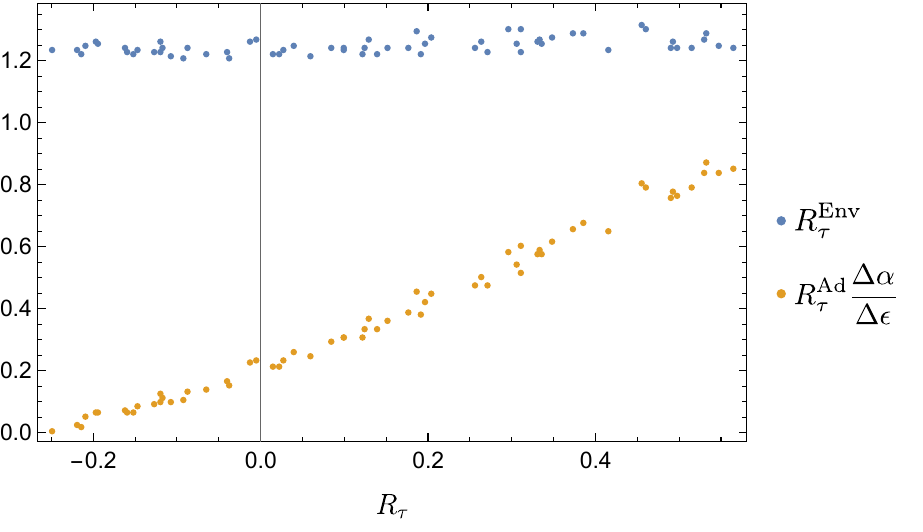}
    \caption{Comparison between the environmental and adaptation components of $R_\tau$ in \eqref{def:Rtau}. In this example, most of the variability in the resilience comes from the adaptation term.}
    \label{fig:RtauParts}
\end{figure}

In Section \ref{sec:Practical} we derived more practical measures to estimate $R$ or $R_\tau$, that we foresee as practical since they can be informed by field data. In Figure \ref{fig:RtauvRs} we evaluate how $\tilde{R}_\tau$ in \eqref{eq:appRtau} and $\hat{R}_\tau$ in \eqref{eq:RtauEnAd} compare to the estimate $R_\tau$. We used the same ensemble and parameter values. The figures suggest that the  empirical versions of these are good approximations for the resilience of the system.

\begin{figure}
   \centering
    \includegraphics[scale=0.7]{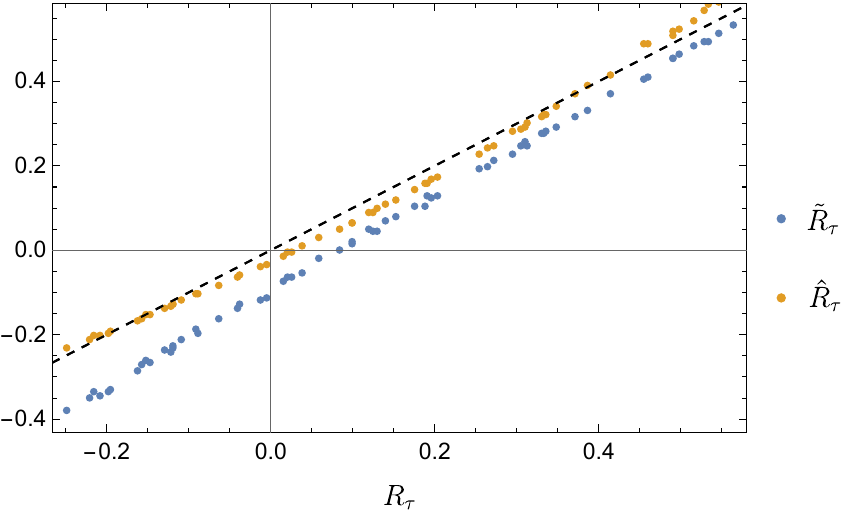}
    \caption{Comparison between $R_\tau$ and its approximations $\tilde{R}_\tau$ (equation \eqref{eq:appRtau}) and $\hat{R}_\tau$ (equation \eqref{eq:RtauEnAd}) obtained through linear response theory. The dashed line marks equality.}
    \label{fig:RtauvRs}
\end{figure}

As mentioned in Section \ref{sec:Practical}, $R_\tau$ can be compared to existing measures or resilience for noisy systems. The value $R^{\VAR}_\tau$ in \eqref{def:Rvar} aims at quantifying resilience by analyzing the changes to the variance in $X$ \citep{ives1995measuring}. Namely, we make $S(x)= S_{\e,\t}(x) = (x- \bar{x}_{\e,\t})^2$ where $\bar{x}_{\e,\t}$ is the mean of the distribution $\bar{p}_{\e,\t}$. Figure \ref{fig:Rtaus}a shows a comparison between $R_\tau$ and $R_{\tau}^{\VAR}$ for all the examples.  Again, since the pattern is increasing, the resilience $R_\tau$ can be used as a proxy for a resilience measure built on the variability induced by the disturbances. This is simply the result of choosing a success function $S$ in \eqref{def:Sexample} that is maximized near the mean $\bar{x}_{\e_0,\t_0}$.

For the model in \eqref{eq:gradientexample}, one can numerically compute the largest eigenvalues $\rho_{\e,\t}$ of the forward operator $\Lop_{\e,\t}^*$ in \eqref{def:LopF} for any $\e,\t$. Thus we can assess whether $R_{\tau}$ gives any information about the rates of exit and return to the equilibrium after the perturbation.  Note that a resilient system would be one in which the eigenvalue after adaptation $|\rho_{\e_1,\t_1}|$ is large compared to the eigenvalue after the environmental perturbation $|\rho_{\e_1,\t_0}|$. Figure \ref{fig:Rtaus}b shows a comparison between $R_\tau$ and the ratio $|\rho_{\e_1,\t_1}|/|\rho_{\e_1,\t_0}|$ for each example, indicating that $R_\tau$ is consistent with this interpretation of resilience.

\begin{figure}
    \centering
    \subfigure[]{\includegraphics[scale=0.59]{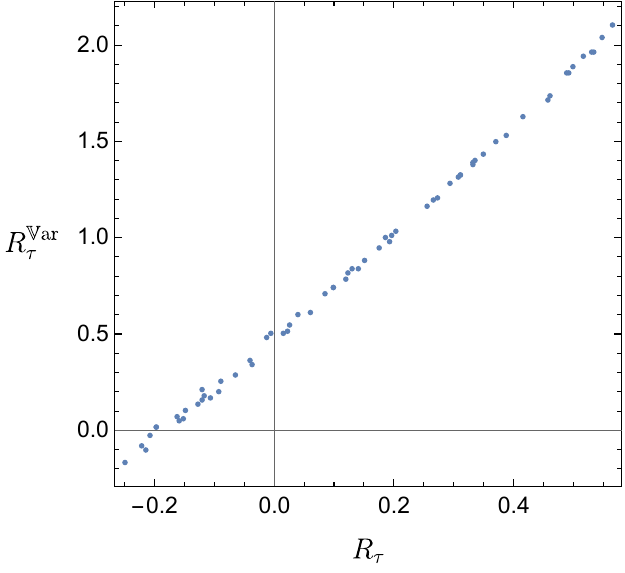}} 
    \subfigure[]{\includegraphics[scale=0.59]{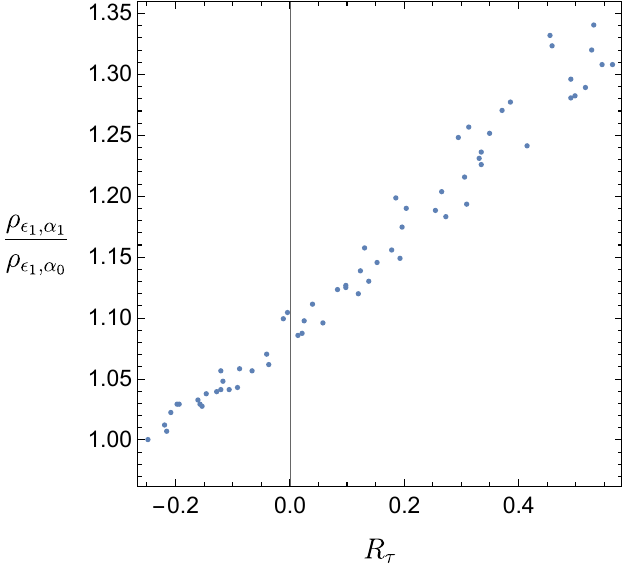}} 
    \caption{(a) Comparison between $R_\tau$ and the measure of resilience $R^{\VAR}_\tau$ in \eqref{def:Rvar} based on the variance of $X_{\e,\t}$. (b) Comparison between $R_\tau$ and the ratio between the dominant eigenvalues of the adapted vs non-adapted system.}  
    \label{fig:Rtaus}
\end{figure}

In summary, the functional $R_\tau$ provides a measure of resilience that, at least for the adaptation strategy used in the examples shown here, is coherent with various notions of resilience used in the context of noisy dynamical systems.

\section{Discussion and Conclusions}
\label{sec:conc}

We have proposed a measure of ecological resilience which quantifies the success of a forced/dissipative system to adapt following an initial applied stress. The relevance of our resilience measure to biological systems, rests upon the assumption that the time evolution of the biological system is described by a stochastic differential equation with a initial stationary probability distribution (prior to the application of a stress). The stochastic nature of the process means that for each starting value of the system there is an ensemble of possible histories (paths). Moreover, our conceptual model also accounts for the random variability among states or subsystems within an organism.

At the mechanistic level, our measure of resilience is strongly inspired by the homeodynamic nature of biological systems that react to some type of imposed stress with an eventual physiological adaptation. Unlike most measures of resilience, we are suggesting that the time history of a well-chosen success observable is essential to the determination of the organism’s resilience to adaptation. Further, we argue that the resilience can be best captured in the ensemble mean among a population of individual subsystems which can have different starting conditions, and random responses to stress. We are suggesting here that a typical success history in noisy/forced systems, as typified by Figure \ref{fig:exampleMeans}, can only be uncovered by ensemble methods such as those proposed in this study. 

It follows from the approximation for $R_\tau$ in \eqref{eq:appRtau} that low resilience will occur in systems with low values of $|\Delta_{\e} \bar{S}^{\tau}_{\e_0,t_0}|$, which in \eqref{eq:LRTtauE} is written as a temporal correlation between functionals of the process $X$. The relationship between  correlations and critical transitions in complex systems is now well-established as in the concept of  `critical slowing down' of correlations \citep{ives1995measuring,Scheffer:2009aa,santos2022}. Indeed, the quantity $\Delta_{\e} \bar{S}^{\tau}_{\e_0,t_0}$  in the example of Section \ref{sec:example} is related to how often system paths transition towards the second stable equilibrium. Resilience, however, must also consider the ability of the system to recover. Our characterization \eqref{eq:RtauParts} of $R_{\tau}^{\rm Env}$ and $ R_{\tau}^{\rm Ad}$ provides the appropriate quantities with respect to which $\Delta_{\e} \bar{S}^{\tau}_{\e_0,t_0}$ must be compared in order obtain a more complete picture.

We now address  the practicality of our definition of resilience to adaptation as  a proposed measure of resilience. We do so by showing how it may be estimated using actual experimental, field or simulation data.

First of all, a dynamic model is not intrinsic to the definition of the resilience measures $R$ or $R_\tau$, and it is {\it not} required to estimate resilience from data. Knowing the densities $\bar{p}_{\e,\t}$ over the state space is also not a requirement. The estimation of $R_\tau$ requires only \textit{samples} from such densities. Namely, to identify, simulate, or prepare an ensemble of individual systems that have been operating under constant environmental conditions, subject them to an environmental disturbance, and measuring the chosen success function $S$ through the adaptation process. The ensemble must be numerous enough to accurately compute expected values of $S$ from sample averages. For estimating $R$, the values of $S(X(t))$ must be sampled from the field sufficiently regularly in time in order to capture (as in Figure \ref{fig:exampleMeans}) the initial, minimum, and limiting values of $\EXP S(X(t))$ required by the definition \eqref{def:usualR}. Our proposal $R_\tau$, on the other hand, requires the detection of the moment at which adaptation occurs in each individual. Measuring success at such moment is the equivalent of sampling the random variable $S(X(\tau))$. In this case only three data points are required from each system: the unperturbed state, the moment of adaptation, and a final state when the population is observed to revert to its new statistical equilibrium.

Linear response theory yields further tools to estimate $R_{\tau}$, as detailed in Section \ref{sec:Estimation}. These estimates take the form of expectations and correlations of specific functionals of the process $X$. See expressions \eqref{eq:S10tau}, \eqref{eq:LRTE1}, \eqref{eq:LRTtauE} and Proposition \ref{prop:gradS}. This estimates can be performed by a combination of experimental data assimilation, simulation and analytical tools, depending on the specific model. A key feature of our framework is that different types of approximation to $R_\tau$ can achieved depending on the level of detail of the available mathematical model. In its full form, every component of the resilience in \eqref{eq:RtauParts} can be estimated by simulating only the unperturbed process $X_{\e_0,\t_0}$.

With regards to the definition of resilience in terms of an almost arbitrary success function $S$, we argue that it yields a more useful and encompassing framework, since notions of stress, well-being, or health can in many case be subjective and application-dependent: $S$ can be  any observable of interest that quantifies the degree to which a system is successful. In principle, $S$ doesn't have to be positive, bounded or continuous for the definition of $R_\tau$ to make sense, however, we suppose that $S \in \text{Dom}(\Lop_{\e,\t})$ in \eqref{def:LopB}, which typically contains only bounded, smooth functions. 
  
 This approach is antithetical to the idea that there is a universal notion of resilience even when confronted with the same systems biology or biological mechanism. For our illustrative example, in \eqref{def:Sexample}, we simply defined $S$ as a non-negative function with a mean centered at an equilibrium point labeled as preferred, with a variance that conveyed how narrowly should success be defined. Alternatively one can propose to measure the resilience of an ensemble of organisms with respect to some biomarker, productivity rate with economic value, or a statistical measure as in the definition \eqref{def:Rvar} in $R_{\tau}^{\VAR}$.

We exemplified our proposal with a very simple model assuming low-dimensional gradient flow dynamics, and forced and adapted the dynamics by simple changes upon the potential. We argue however, that our approach is amiable to biological systems that not only change through time, but have inherent hierarchies of biological organization. For real-life  biological systems assessing resilience will involve a series of interconnected dynamics described by mathematical or data-driven models ({\it e.g.}, modeling cell to leaf,  leaf to plant, etc). The more complex model might capture the exact relationship between the state variable and the external forcing, or the dynamics of adaptation itself. Our proposal for  resilience to adaptation will go through on these more complex systems, so long as assumptions regarding the stationarity of distribution of the state vector dynamics hold.

The idea of using ensembles and a success function as the means to estimate resilience, is very general. It is, however, limited to situations in which samples from the same distribution can be measured, simulated or observed, and some control over the environmental perturbations can be exerted. The estimation of $S(X(\tau))$ in \eqref{def:Rtau} brings the additional challenge of requiring a way to detect or model when a threshold into adaptation has occurred. The particular model \eqref{eq:model} has some underlying assumptions that limit the scope of this work. For one, the smoothness and regularity conditions on $F$ might not hold, but more significantly the existence of an invariant distribution might be unknown or impossible to assume. The assumption of linear noise is also problematic in those systems where environmental perturbations increase inherent variability. However, although not explicit in this paper, linear response theory can also  be applied to perturbations on $\sigma = \sigma_{\e,\t}$ \citep[see][for example]{pavliotis}. The estimation tools provided in section \ref{sec:Estimation} impose more specific limitations which might not hold in general. Specifically, linear respons theory requires small $\Delta\e$ and $\Delta \t$ and the smoothness of $F_{\e,\t}$ and $\bar{p}_{\e,\t}$ with respect to $\e$ and $\t$.

\bibliographystyle{apalike}

\bibliography{bio}

\begin{thebibliography}{}

\bibitem[Angeler and Allen, 2016]{angeler2016quantifying}
Angeler, D.~G. and Allen, C.~R. (2016).
\newblock Quantifying resilience.
\newblock {\em Journal of Applied Ecology}, 53(3):617--624.

\bibitem[Arani et~al., 2021]{arani2021exit}
Arani, B.~M., Carpenter, S.~R., Lahti, L., Van~Nes, E.~H., and Scheffer, M.
  (2021).
\newblock Exit time as a measure of ecological resilience.
\newblock {\em Science}, 372(6547):eaay4895.

\bibitem[Arnoldi et~al., 2016]{arnoldi2016resilience}
Arnoldi, J.-F., Loreau, M., and Haegeman, B. (2016).
\newblock Resilience, reactivity and variability: A mathematical comparison of
  ecological stability measures.
\newblock {\em Journal of theoretical biology}, 389:47--59.

\bibitem[Boettner and Boers, 2022]{boettner22}
Boettner, C. and Boers, N. (2022).
\newblock Critical slowing down in dynamical systems driven by nonstationary
  correlated noise.
\newblock {\em Physical Review Research}, 4:013230.

\bibitem[Brenig, 2012]{brenig2012statistical}
Brenig, W. (2012).
\newblock {\em Statistical theory of heat: nonequilibrium phenomena}.
\newblock Springer Science \& Business Media.

\bibitem[Dakos and K{\'e}fi, 2022]{Dakos22}
Dakos, V. and K{\'e}fi, S. (2022).
\newblock Ecological resilience: what to measure and how.
\newblock {\em Environmental Research Letters}, 17:043003.

\bibitem[Fatichi et~al., 2014]{fatichi2014moving}
Fatichi, S., Leuzinger, S., and K{\"o}rner, C. (2014).
\newblock Moving beyond photosynthesis: from carbon source to sink-driven
  vegetation modeling.
\newblock {\em New Phytologist}, 201(4):1086--1095.

\bibitem[Fatichi et~al., 2016]{fatichi2016modeling}
Fatichi, S., Pappas, C., and Ivanov, V.~Y. (2016).
\newblock Modeling plant--water interactions: an ecohydrological overview from
  the cell to the global scale.
\newblock {\em Wiley Interdisciplinary Reviews: Water}, 3(3):327--368.

\bibitem[Freidlin and Wentzell, 1998]{freidlin1998}
Freidlin, M.~I. and Wentzell, A.~D. (1998).
\newblock {\em Random perturbations of dynamical systems}.
\newblock Springer, New York.

\bibitem[Gardiner, 2004]{gardiner}
Gardiner, C.~W. (2004).
\newblock {\em Handbook of Stochastic Methods}.
\newblock Springer, Berlin.

\bibitem[Guckenheimer and Holmes, 1983]{guck83}
Guckenheimer and Holmes, P. (1983).
\newblock {\em Nonlinear Oscillations, Dynamical Systems, and Bifurcations of
  Vector Fields}.
\newblock Springer.

\bibitem[Guti\'errez and Lucarini, 2022]{santos2022}
Guti\'errez, M.~S. and Lucarini, V. (2022).
\newblock On some aspects of the response to stochastic and deterministic
  forcings.
\newblock {\em Journal of Physics A: Mathematical and Theoretical},
  55(42):425002.

\bibitem[Hairer and Majda, 2010]{Hairer2010}
Hairer, M. and Majda, A.~J. (2010).
\newblock {A simple framework to justify linear response theory}.
\newblock {\em Nonlinearity}, 23(4):909--922.

\bibitem[Held and Kleinen, 2004]{Held2004}
Held, H. and Kleinen, T. (2004).
\newblock Detection of climate system bifurcations by degenerate
  fingerprinting.
\newblock {\em Geophysical Research Letters}, 31(23).

\bibitem[Holling, 1973]{Holling73}
Holling, C. (1973).
\newblock Resilience and stability of ecological systems.
\newblock {\em Annual Review of Ecological Systems}, 4:1--23.

\bibitem[Holling, 1996]{holling96}
Holling, C. (1996).
\newblock Engineering resilience versus ecological resilience.
\newblock {\em Engineering within ecological constraints}, 31.

\bibitem[Ives, 1995]{ives1995measuring}
Ives, A.~R. (1995).
\newblock Measuring resilience in stochastic systems.
\newblock {\em Ecological Monographs}, 65(2):217--233.

\bibitem[Krakovsk{\'a} et~al., 2022]{Krakovska16}
Krakovsk{\'a}, H., K{\"u}hn, C., and Longo, I. (2022).
\newblock Resilience of dynamical systems.
\newblock {\em arXiv}.
\newblock 2105.10512v2.

\bibitem[Lenton et~al., 2008]{Lenton2008}
Lenton, T., Held, H., Kriegler, E., Hall, J., Lucht, W., Rahmstorf, S., and
  Schellnhuber, H. (2008).
\newblock {Tipping elements in the {\{}Earth's climate system{\}}}.
\newblock {\em Proc. Natl. Acad. Sci. USA}, 105:1786--1793.

\bibitem[Lloyd et~al., 2001]{lloyd2001homeodynamics}
Lloyd, D., Aon, M.~A., and Cortassa, S. (2001).
\newblock Why homeodynamics, not homeostasis?
\newblock {\em TheScientificWorldJournal}, 1:133--145.

\bibitem[Lucarini, 2016]{Lucarini2016}
Lucarini, V. (2016).
\newblock {Response Operators for Markov Processes in a Finite State Space:
  Radius of Convergence and Link to the Response Theory for Axiom A Systems}.
\newblock {\em Journal of Statistical Physics}, 162(2):312--333.

\bibitem[Meyer, 2016]{Meyer16}
Meyer, K. (2016).
\newblock A mathematical review of resilience in ecology.
\newblock {\em Natural Resource Modeling}, 29:339--355.

\bibitem[Pavliotis, 2014]{pavliotis}
Pavliotis, G. (2014).
\newblock {\em Stochastic Processes and Applications}.
\newblock Springer.

\bibitem[Ruelle, 2009]{Ruelle2009}
Ruelle, D. (2009).
\newblock {A review of linear response theory for general differentiable
  dynamical systems}.
\newblock {\em Nonlinearity}, 22(4):855--870.

\bibitem[Scheffer et~al., 2009]{Scheffer:2009aa}
Scheffer, M., Bascompte, J., Brock, W., Brovkin, V., S.R.~Carpenter, V. D.~V.,
  Held, H., Nes, E.~V., Rietkerk, M., and Sugihara, G. (2009).
\newblock Scheffer, marten, et al. "early-warning signals for critical
  transitions.
\newblock {\em Nature}, 461:53--59.

\bibitem[Shang, 2023a]{shang2023matrix}
Shang, Y. (2023a).
\newblock Matrix-scaled resilient consensus of discrete-time and
  continuous-time networks.
\newblock {\em Quarterly of Applied Mathematics}, 81(4):777--800.

\bibitem[Shang, 2023b]{shang2023resilient}
Shang, Y. (2023b).
\newblock Resilient vector consensus over random dynamic networks under mobile
  malicious attacks.
\newblock {\em The Computer Journal}, page bxad043.

\bibitem[Strogatz, 2018]{strogatz2018nonlinear}
Strogatz, S.~H. (2018).
\newblock {\em Nonlinear dynamics and chaos with student solutions manual: With
  applications to physics, biology, chemistry, and engineering}.
\newblock CRC press.

\bibitem[Van~Meerbeek et~al., 2021]{van2021unifying}
Van~Meerbeek, K., Jucker, T., and Svenning, J.-C. (2021).
\newblock Unifying the concepts of stability and resilience in ecology.
\newblock {\em Journal of Ecology}, 109(9):3114--3132.

\bibitem[Willey, 2018]{willey2018environmental}
Willey, N. (2018).
\newblock {\em Environmental plant physiology}.
\newblock Garland Science.

\bibitem[Yi and Jackson, 2021]{yi2021review}
Yi, C. and Jackson, N. (2021).
\newblock A review of measuring ecosystem resilience to disturbance.
\newblock {\em Environmental Research Letters}, 16(5):053008.

\end{thebibliography}

\listoffigures
\end{document}